\pdfoutput=1
\documentclass[12pt,a4paper]{article}

\usepackage[utf8]{inputenc}
\usepackage{amsmath, amssymb, amsthm}
\usepackage{booktabs}
\usepackage{makecell}
\usepackage{longtable}
\usepackage{graphicx}
\usepackage{authblk}
\usepackage{rotating}
\usepackage{adjustbox}
\usepackage{float}
\usepackage{bm}
\usepackage{xspace}
\usepackage[authoryear,round]{natbib}
\usepackage[margin=1in]{geometry}
\usepackage{hyperref}
\usepackage{enumitem}
\usepackage[capitalise,noabbrev]{cleveref}

\hypersetup{
  colorlinks=true,
  linkcolor=blue,
  citecolor=blue,
  urlcolor=blue,
  pdfauthor={Satyam Tyagi},
  pdftitle={Product--Congruence Games: A Unified Impartial-Game Framework for RSA (phi-MuM) and AES (poly-MuM)},
  pdfkeywords={product-congruence games, multiplicative modular Nim, RSA, AES, Sprague--Grundy theory, Threshold Region, Indeterminacy Region, operation alignment, aggregation--compression}
}

\theoremstyle{plain}
\newtheorem{theorem}{Theorem}[section]
\newtheorem{proposition}[theorem]{Proposition}
\newtheorem{lemma}[theorem]{Lemma}
\newtheorem{corollary}[theorem]{Corollary}

\theoremstyle{definition}
\newtheorem{definition}[theorem]{Definition}
\newtheorem{example}[theorem]{Example}

\theoremstyle{remark}
\newtheorem{remark}{Remark}

\newtheorem{principle}[theorem]{Principle}

\newcommand{\GF}{\mathbb{F}}
\newcommand{\ZZ}{\mathbb{Z}}
\newcommand{\NN}{\mathbb{N}}
\newcommand{\PCG}{\ensuremath{\mathrm{PCG}}\xspace}
\newcommand{\polyMuM}{\ensuremath{\textit{poly}\text{--}\mathrm{MuM}}\xspace}
\newcommand{\phiMuM}{\ensuremath{\phi\text{--}\mathrm{MuM}}\xspace}
\newcommand{\SG}{\ensuremath{\operatorname{SG}}\xspace}
\newcommand{\mex}{\operatorname{mex}}
\newcommand{\Opt}{\operatorname{Opt}}
\newcommand{\ord}{\operatorname{ord}}

\newcommand{\GFtwoeight}{\ensuremath{\GF_{2^{\!8}}}\xspace} 
\newcommand{\GFtwoeightstar}{\ensuremath{\GF_{2^{\!8}}^\times}\xspace}
\newcommand{\GFq}{\ensuremath{\GF_q}\xspace}
\newcommand{\GFqstar}{\ensuremath{\GF_q^\times}\xspace}

\title{Product--Congruence Games: A Unified Impartial-Game Framework for RSA (\phiMuM) and AES (\polyMuM)}

\author{Satyam Tyagi\thanks{Email: \href{mailto:satyam.tyagi@gmail.com}{satyam.tyagi@gmail.com}, ORCID: \href{https://orcid.org/0009-0001-0137-2338}{0009-0001-0137-2338}}}
\affil{Independent Researcher}

\date{\today}

\begin{document}
\maketitle

\begin{abstract}
  \textbf{RSA exponent reduction and AES S-box inversion share a hidden commonality: both are governed by the same impartial combinatorial principle, which we call a \emph{Product--Congruence Game}~(\PCG).} A \PCG tracks play via the modular (or finite-field) product of heap values, providing a single invariant that unifies the algebraic cores of these two ubiquitous symmetric and asymmetric cryptosystems.
  
  We instantiate this framework with two companion games. First, \textbf{\phiMuM}, where a left-associated ``multi-secret'' RSA exponent chain compresses via our \emph{Compression Theorem} into the game of Multiplicative Modular Nim, \(\PCG(k,\{1\})\), with \(k=\ord_N(g)\). The losing predicate of this game can then be factorized using the Chinese Remainder Theorem, mirroring RSA's implementation. Second, \textbf{\polyMuM}, our model for finite-field inversion in systems like AES. For \polyMuM, we prove the \emph{single-hole property} holds within its Threshold Region, which in turn implies that the Sprague-Grundy (\SG) values are multiplicative for disjunctive sums in that region.
  
  Beyond these concrete instances, we establish four structural theorems for general Product--Congruence Games \(\PCG(m,R)\): (i) single-heap repair above the modulus, (ii) ultimate period \(m\) per coordinate, (iii) exact and asymptotic losing densities, and (iv) confinement of optimal play to a finite \emph{Indeterminacy Region}. An \emph{operation-alignment} collapse principle finally explains why some variants degenerate to a single aggregate while MuM,\phiMuM and \polyMuM retain non-trivial structure.
  
  All ingredients (orders, CRT, finite fields) are classical; the novelty lies in this unified aggregation-compression viewpoint that embeds RSA and AES inside one impartial-game framework, together with the structural and collapse theorems for the resulting games.
\end{abstract}


\bigskip
\noindent\textbf{MSC Classification (2020):} 91A46 (primary), 11A07, 94A60, 05A99


\section{Introduction}\label{sec:intro}

Our \textbf{main contribution} is to show that the algebraic cores of
\textbf{RSA} (exponent reduction via Carmichael and the Chinese Remainder
Theorem) and \textbf{AES} (finite-field inversion in \GFtwoeight) can be analysed
inside the \emph{same class of impartial combinatorial games}.  The key is a
single multiplicative \emph{product-congruence} invariant.  A left-associated
``multi-secret RSA exponent chain'' compresses exactly to Multiplicative Modular
Nim (MuM); its finite-field analogue, \polyMuM, models the AES
S-box inversion.

Impartial games often admit succinct aggregating invariants (Nim's bitwise XOR
\citep{Bouton1901,SpragueGrundy}).  Our earlier work introduced MuM, replacing
XOR by a modular product \citep{Tyagi2025}.  Here we show that the \emph{same}
multiplicative invariant explains both RSA's exponent reduction and AES's field
inversion when they are cast as \emph{product-congruence games}.

\paragraph{From exponent chain to product congruence (Compression Theorem).}
Fix an RSA modulus $N$ and $g\in(\ZZ/N\ZZ)^\times$.  A position
$\mathbf h=(h_1,\dots,h_n)$ evaluates to
\[
  E(\mathbf h)=(((g^{h_1})^{h_2})^{h_3}\cdots)^{h_n}\pmod N,
\]
and is losing iff $E(\mathbf h)=g$.  Left-association gives
$E(\mathbf h)=g^{H}$ with $H=\prod_i h_i$.  If $k=\ord_N(g)$ then
\[
  E(\mathbf h)\equiv g
  \;\Longleftrightarrow\;
  H\equiv 1 \pmod{k},
\]
so the chain game is isomorphic to $\PCG(k,\{1\})$ (MuM at modulus $k$).  With
$k=\prod p_j^{\alpha_j}$ the losing predicate CRT-factorizes, mirroring RSA
implementations.

\paragraph{A unified framework: Product-Congruence Games.}
The compressed form motivates \textbf{product-congruence games}
$\PCG(m,R)$: positions are integer vectors, a move strictly decreases one heap,
and a position is losing when $\prod t_i\bmod m\in R$.  The chain game is
$\PCG(k,\{1\})$; classical MuM is $\PCG(m,\{1\})$.  The finite-field instance
\polyMuM plays the same game over \GFtwoeightstar, matching the AES
S-box inversion step.

\paragraph{Structural localisation and the Indeterminacy Region.}
For every $\PCG(m,R)$ we prove:
(i) \textbf{single-heap repair} - if the residue is not in $R$ and some heap
$\ge m$, one move reaches $R$;
(ii) \textbf{ultimate period~$m$} in each coordinate (numeric model);
(iii) confinement of all residual strategy to the finite
\emph{Indeterminacy Region} $I_m=[1,m-1]^n$ (field analogue: $I_q$ with
$q-1$).  The proof uses an \emph{aggregation-compression} normalisation: aggregate
the invariant, compress to one canonical heap, then perform the single legal
repair move.

\paragraph{Collapse boundaries (operation alignment).}
If the local move operation aligns with the invariant (e.g.\ subtraction for a
sum-mod game, or divisor moves for a product-mod game) the multi-heap game
collapses outside the finite region to a trivial single aggregate.  MuM and
\polyMuM avoid collapse because their moves are additive while the
invariant is multiplicative.

\subsection*{Key Contributions}
\begin{enumerate}
  \item \textbf{Unified RSA/AES game lens.}  RSA's exponent reduction and the
        AES S-box inversion live in the \emph{same} impartial-game class driven
        by a product-congruence invariant (MuM / \polyMuM).
  \item \textbf{Compression Theorem.}  A multi-secret RSA exponent chain is
        isomorphic to $\PCG(k,\{1\})$ with $k=\ord_N(g)$; the losing predicate
        factorizes via CRT.
  \item \textbf{Structural theorems for $\PCG(m,R)$.}  \emph{SG
        multiplicativity in the Threshold Region~$T_\bullet$} (numeric and
        field models) via the single-hole property; single-heap repair;
        ultimate period~$m$ per coordinate (numeric); exact/asymptotic losing
        densities; all residual complexity confined to the finite
        $I_\bullet$.
  \item \textbf{Collapse boundaries.}  An operation-alignment principle
        explains why aligned variants collapse while MuM and
        \polyMuM retain strategic depth.
\end{enumerate}

\subsection*{Related Work}
Orders, Carmichael's $\lambda$, CRT \citep{rsa78,Carmichael}, AES field
inversion \citep{fips197}, and impartial game theory
\citep{SpragueGrundy} are classical.  Our earlier paper introduced MuM
and~\polyMuM \citep{Tyagi2025}.  Strategic game-theoretic models of
cryptography do exist—see Katz's survey on rational protocols
\citep{KatzGameCrypto2008}—but they analyse \emph{mixed-strategy}
equilibria for secure-multiparty computation and rely on probabilistic
moves, rather than on the \emph{perfect-information impartial games}
used here.  We are not aware of prior work that places RSA and AES
under a single impartial-game invariant or proves the accompanying
compression and collapse theorems.

\paragraph{Organisation.}
\begin{itemize}[leftmargin=2em,itemsep=0pt]
  \item \cref{sec:chain} proves the Compression Theorem.
  \item \cref{sec:pcg} formalises \(\PCG(m,R)\) and its structure.
  \item \cref{sec:collapse} gives collapse boundaries.
  \item \cref{sec:crypto} develops the unified RSA (\phiMuM) / AES (\polyMuM) perspective.
  \item \cref{sec:conclusion} concludes with future directions.
\end{itemize}

\section{The Exponent Chain Game and its Compression}\label{sec:chain}

This section isolates the RSA lens (the ``exponent-chain'' game) and shows it
\emph{compresses} exactly to MuM at modulus \(k=\ord_N(g)\).

\subsection{Game definition}
Fix an RSA modulus \(N\ge 2\) and \(g\in(\ZZ/N\ZZ)^\times\) of order
\(k=\ord_N(g)\) (so \(k\mid\lambda(N)\)).  The exponent-chain
(\emph{``\phiMuM''}) game is:

\begin{description}
\item[Positions] \(\mathbf h=(h_1,\dots,h_n)\in\NN^n\) with each \(h_i\ge 1\).
\item[Evaluation]
  \[
    E(\mathbf h)=(((g^{h_1})^{h_2})^{h_3}\cdots)^{h_n}\pmod N.
  \]
\item[Losing predicate] \(\mathbf h\) is losing iff \(E(\mathbf h)\equiv g\pmod N\).
\item[Moves] Pick an index \(i\) with \(h_i>1\) and replace \(h_i\) by some
  \(h_i'<h_i\).  \emph{Once \(h_i\ge k\) we restrict the decrement to
  \(1\le h_i-h_i'\le k-1\)} to forbid residue-preserving “null moves.”
\end{description}

\paragraph{Flattening.}
Because \((g^a)^b=g^{ab}\) we have
\[
  E(\mathbf h)\equiv g^{H(\mathbf h)}
  \quad\text{with}\quad
  H(\mathbf h):=\prod_{i=1}^n h_i.
\]

\subsection{Compression and CRT factorisation}

\begin{theorem}[Compression Theorem]\label{thm:compression}
\(E(\mathbf h)\equiv g\pmod N\) \emph{iff} \(H(\mathbf h)\equiv 1\pmod{k}\).
\end{theorem}

\begin{proof}
  Flattening the exponent chain gives
  \(E(\mathbf h) \equiv g^{H(\mathbf h)} \pmod N\).
  By the definition of multiplicative order \(k=\ord_N(g)\),
  \begin{align*}
      E(\mathbf h) \equiv g \pmod N &\iff g^{H(\mathbf h)} \equiv g^1 \pmod N \\
                                    &\iff g^{H(\mathbf h)-1} \equiv 1 \pmod N \\
                                    &\iff k \mid \bigl(H(\mathbf h)-1\bigr) \\
                                    &\iff H(\mathbf h) \equiv 1 \pmod k. \qedhere
  \end{align*}
\end{proof}

\begin{corollary}[CRT factorisation]\label{cor:crt}
If \(k=\prod_{j=1}^r p_j^{\alpha_j}\), then \(\mathbf h\) is losing iff
\(H(\mathbf h)\equiv 1\pmod{p_j^{\alpha_j}}\) for all \(j\).
\end{corollary}

\begin{remark}[Strategic isomorphism to MuM]
\cref{thm:compression} shows the chain game is strategically isomorphic to
\(\PCG(k,\{1\})\) (MuM at modulus \(k\)): losing means
“\(\prod h_i\equiv1\pmod{k}\).”  Bounding each decrement by \(<k\) is
\emph{exactly} the MuM move rule and prevents losing\(\to\)losing null moves.
\end{remark}

\begin{example}
Let \(N=15\) and \(g=2\).  Then \(k=\ord_{15}(2)=4\).
A position \((h_1,h_2)\) is losing iff \(h_1h_2\equiv1\pmod4\); e.g.\
\((1,1)\), \((3,3)\), \((5,1)\).
\end{example}

The \emph{structure} of the compressed game (single-heap repair, Threshold vs.\
Indeterminacy, SG behaviour, densities, periodicity) is developed once and for
all in the next section, inside the general PCG framework that also covers the
AES/\polyMuM case.

\section{Product--Congruence Games (PCG): the unified framework}
\label{sec:pcg}

We now introduce the game class that simultaneously models

\begin{itemize}[leftmargin=2em,itemsep=0pt]
  \item classical MuM (\(\PCG(m,\{1\})\));
  \item the compressed RSA exponent-chain game (``\phiMuM'');
  \item the AES field variant (\polyMuM).
\end{itemize}

\subsection{Definition (ring / field form)}
\label{subsec:pcg-def}

\begin{definition}[PCG over a finite ring]\label{def:pcg-ring}
Let \(R\) be a finite commutative ring with unit group \(R^\times\).
Let \((S,<)\) be a well-ordered set of heap labels
(often \(S\subset\NN\)), and let
\(\psi:S\to R^\times\) be an \emph{injective} (canonical) encoding.
Fix a non-empty losing set \(L\subseteq R^\times\).

\begin{description}
  \item[Position] \(\mathbf h=(h_1,\dots,h_n)\in S^n\).
  \item[Invariant] \(\displaystyle
        \phi(\mathbf h):=\prod_{i=1}^n\psi(h_i)\in R^\times.\)
  \item[Move] Choose \(h_j\) and replace it by
        \(h_j'<h_j\) (so \(\psi(h_j')\neq\psi(h_j)\)).
  \item[Losing] \(\mathbf h\) is losing iff \(\phi(\mathbf h)\in L\).
\end{description}
We denote the game by \(\PCG(R,S,\psi,L)\).
\end{definition}

\paragraph{Two cryptographic lenses (our focus).}
\begin{itemize}[leftmargin=2em,itemsep=2pt]
  \item \textbf{\phiMuM (RSA).}  
        \(R=\ZZ/k\ZZ\) with \(k=\ord_N(g)\); \(S=\{1,2,\dots\}\);
        \(\psi(h)=h\bmod k\); \(L=\{1\}\).
  \item \textbf{\polyMuM (AES).}  
        \(R=\) \GFq (\(q=p^n\));
        \(S=\{1,\dots,q-1\}\); \(\psi=s\) (canonical bijection onto
        \GFqstar); \(L=\{1\}\).
\end{itemize}

For clarity the structural proofs below use the \emph{numeric} instance
\(\PCG(m,R)\) with \(m\ge2\); the field results are identical with
\(m\) replaced by \(q-1\) and \(S=[1,q-2]\).

\subsection{Threshold vs.\ finite Indeterminacy Region}
\label{sec:pcg-split}
  
\begin{definition}[Regions]\label{def:regions}
For \(\PCG(m,R)\) put
\[
  I_m := [1,m-1]^n
  \quad\text{(finite Indeterminacy Region)},\qquad
  T_m := \NN^n\setminus I_m
  \quad\text{(Threshold Region)}.
\]
\end{definition}

If some heap \(t_j\ge m\) then \(t_j,t_j-1,\dots,t_j-(m-1)\) cover every
residue mod \(m\), which drives the next two lemmas.

\begin{lemma}[One-move repair in \(T_m\)]\label{lem:repair}
Let \(t\in T_m\) with
\(\Pi(t):=\prod_i t_i \notin R\).
Then there is a move, acting on a heap \(\ge m\), that produces
\(t'\) with \(\Pi(t')\in R\).
\end{lemma}

\begin{proof}
Pick \(t_j\ge m\) and set
\(C=\prod_{i\ne j} t_i\bmod m\in(\ZZ/m\ZZ)^\times\).
For some \(r\in R\) choose \(d\in[0,m-1]\) with
\(t_j-d\equiv C^{-1}r\pmod m\); this \(d\) exists because the residue
window is complete.  Replacing \(t_j\) by \(t_j-d\) yields
\(\Pi(t')\equiv r\pmod m\).
\end{proof}

\begin{lemma}[Losing-move blocking in \(T_m\)]\label{lem:block}
If \(t\in T_m\) and \(\Pi(t)\in R\), then every legal move gives
\(\Pi(t')\notin R\).
\end{lemma}

\begin{proof}
With \(t_j\ge m\) all residues are reachable; hence no move can keep the
(unique) losing residue in place.
\end{proof}

\begin{proposition}[Outcome classification on \(T_m\)]\label{prop:Tm-outcome}
For \(t\in T_m\),
\[
  t\in\mathcal P \;\Longleftrightarrow\; \Pi(t)\in R,
  \qquad
  t\in\mathcal N \;\Longleftrightarrow\; \Pi(t)\notin R.
\]
\end{proposition}

\begin{proof}
Combine Lemmas \ref{lem:repair} and \ref{lem:block}.
\end{proof}

\subsection{Single-hole property and SG multiplicativity on \(T_m\)}

\begin{definition}[Single-hole property]\label{def:SP}
A position \(G\) has SP if the set
\(\{\SG(G'):G'\in\Opt(G)\}\)
omits exactly one value of the SG domain.
\end{definition}

\begin{proposition}[SP in \(T_m\)]\label{prop:SP}
Every component game in \(T_m\) has SP.
\end{proposition}

\begin{proof}
From a maximal heap (\(\ge m\)) one can reach every smaller residue class,
so the option SG values sweep all but the current one.
\end{proof}

\begin{theorem}[SG multiplicativity on \(T_m\)]\label{thm:sg-mult}
For \(G_1,G_2\in T_m\) in \(\PCG(m,\{1\})\),
\[
  \SG(G_1\oplus G_2) \equiv
  \SG(G_1)\,\SG(G_2)\pmod m.
\]
\end{theorem}

\begin{proof}
Both components satisfy SP; their option sets are complements of single
values \(a,b\).  Options of \(G_1\oplus G_2\) miss exactly \(ab\pmod m\);
\(\mex\) (with the fixed indexing) returns \(ab\).
\end{proof}

\begin{remark}\label{rem:inside-Im}
Inside \(I_m\) a heap \(<m\) does not see all residues; SP can fail and the
multiplicative law need not hold.  The field case behaves identically with
\(T_q/I_q\).
\end{remark}

\subsection{Normalisation (aggregation-compression) in \(I_m\)}

Write \(A(t)=\Pi(t)\bmod m\).
Define the analytic normalisation
\[
  \mathrm{Norm}(t):=(A(t)),
\]
i.e.\ compress all heaps to the single canonical heap carrying the product.

\begin{lemma}[Normalisation preserves the invariant]\label{lem:norm}
\(A(\mathrm{Norm}(t))=A(t)\).  Hence \(t\) is losing iff
\(\mathrm{Norm}(t)\) is losing.
\end{lemma}

\begin{lemma}[Normalisation \(+\) one legal repair]\label{lem:norm-repair}
If \(t\in I_m\) and \(A(t)\notin R\) then from \(\mathrm{Norm}(t)\)
a single legal move reaches \(R\).
\end{lemma}

\begin{proof}
Apply \cref{lem:repair} to the one-heap normalised position.
\end{proof}

The field variant is identical with \(m\) replaced by \(q-1\) and the
canonical representative map \(C\) in place of the identity.

\subsection{Ultimate periodicity (numeric) and densities}

\begin{theorem}[Ultimate periodicity]\label{thm:period}
Fix \(t_{-j}\) and view the outcome as a function of one heap \(x\).
Then \(f_j(x)\) is ultimately periodic with period \(m\):
for all \(x\ge m\), \(f_j(x+m)=f_j(x)\).
\end{theorem}

\begin{proof}
For \(x\ge m\) the position lies in \(T_m\); by
\cref{prop:Tm-outcome} its outcome depends only on
\(\Pi(t_{-j})\,x\bmod m\), hence only on \(x\bmod m\).
\end{proof}

\begin{proposition}[Exact / asymptotic losing densities]\label{prop:dens}
In finite-group variants (field \polyMuM with \(|G|=q-1\); RSA
\phiMuM with \(|G|=k\)) the losing density is exactly
\(1/|G|\).  For unbounded integer MuM the density tends to \(1/m\) (or
\(1/\varphi(m)\) if restricted to units).
\end{proposition}

\begin{proof}
Fix \(n-1\) coordinates; the last heap is uniquely determined to hit the
losing residue, giving exactly \(|G|^{\,n-1}\) losing positions among
\(|G|^{\,n}\).  In the unbounded model, residues become asymptotically
independent and uniform.
\end{proof}

\section{Collapse Boundaries and Operation Alignment}\label{sec:collapse}

Whenever the \emph{one-heap move set} on a sufficiently large heap can generate
\emph{all} target congruence classes of the invariant, the multi-heap game
collapses (outside a finite kernel) to a single-heap problem on the aggregate
invariant.  We illustrate this with two toy games, then state a general
operation-alignment principle, and finally explain why MuM/\polyMuM
avoid a \emph{total} collapse.

\subsection{Additive Sum-Congruence Game}\label{sec:collapse-additive}

Fix \(m\ge2\).  Positions are
\(\mathbf t=(t_1,\dots,t_n)\in\NN_{\ge1}^n\),
the losing predicate is
\[
  \sum_{i=1}^n t_i \equiv s \pmod m,
\]
and a move replaces any heap by any strictly smaller positive integer.  Define

\[
\begin{aligned}
  I_m^{+} &:= [1,m-1]^n,
  &\qquad&\text{(finite region).}
\end{aligned}
\]

\begin{theorem}[Additive collapse]\label{thm:additive-collapse}
If some \(t_j \ge m\), then by moving on heap \(j\) alone one can reach
\emph{every} residue class for the total sum modulo \(m\) in a single move
(in particular the losing residue \(s\)).  Hence every position outside
\(I_m^{+}\) is equivalent to its residue-class representative; the multi-heap
structure is trivial (a one-heap game) there.
\end{theorem}

\begin{proof}
Let \(S=\sum_i t_i\) and choose \(j\) with \(t_j\ge m\).
For any target residue \(r\pmod m\) select
\(d\in\{0,\dots,m-1\}\) with \(S-d\equiv r\pmod m\),
set \(t_j':=t_j-d\ge1\), and observe that the new sum is \(S-d\equiv r\).
\end{proof}

\subsection{Divisor-Move Product Game}\label{sec:collapse-divisor}

Fix \(m\ge2\).  Positions are
\(\mathbf t=(t_1,\dots,t_n)\in\NN_{\ge1}^n\),
the invariant is
\[
  P(\mathbf t)=\prod_{i=1}^n t_i \pmod m,
\]
and a move replaces one heap \(t_j\) by any proper divisor \(d\mid t_j\),
\(d<t_j\).  Let \(G:=(\ZZ/m\ZZ)^\times\) and denote \(\bar x:=x\bmod m\).

\[
\begin{aligned}
  I_m^{\times}
  &:= \bigl\{\mathbf t:\forall j,\,
        t_j<m \text{ or } \gcd(t_j,m)\neq1\bigr\},
  &\qquad&\text{(no heap is both large and a unit).}
\end{aligned}
\]

\begin{theorem}[Divisor collapse]\label{thm:divisor-collapse}
Assume there exists \(M\) such that for every \(t\ge M\) with
\(\gcd(t,m)=1\) the set
\[
  \Bigl\{\overline{d}\,\overline{t}^{-1}\in G : d\mid t,\; d<t\Bigr\}
\]
generates \(G\).
Then every position outside \(I_m^{\times}\) is strongly equivalent to the
single-heap game holding the aggregate product \(P(\mathbf t)\):
one heap alone can move \(P(\mathbf t)\) to any unit residue (in particular
the losing one), and all further play can be simulated on that aggregate.
\end{theorem}

\begin{proof}
Pick a heap \(t_j\ge M\) with \(\gcd(t_j,m)=1\).
Replacing \(t_j\) by a divisor \(d\) multiplies \(P(\mathbf t)\) by the unit
\(\overline{d}\,\overline{t_j}^{-1}\in G\).
By hypothesis these units generate \(G\), so heap \(j\) alone can realise any
target residue.  Once \(P(\mathbf t)\) is set, the remaining coordinates are
strategically irrelevant.
\end{proof}

\begin{remark}
The hypothesis is a clean \emph{operation-alignment} condition: the
local move semigroup generated by one large heap acts transitively on the target
group \(G\).  Without it, collapse need not occur.
\end{remark}

\subsection{A general operation-alignment principle}

Both toy games exemplify the following abstraction.

Let \((M,\circ)\) be a cancellative commutative monoid and
\(\pi:M\to Q\) a surjective homomorphism onto a finite monoid \(Q\) that contains
a losing set \(L\subseteq Q\).
A position is \(\mathbf h=(h_1,\dots,h_n)\) with \(h_i\in M\); the invariant is
\(\phi(\mathbf h)=\pi(\circ_i h_i)\).

\begin{principle}[Operation alignment]\label{prin:alignment}
Suppose there is a finite kernel \(B\subset M\) such that for every
\(x\in M\setminus B\) the set of one-heap moves from \(x\) maps (under \(\pi\))
to a subset that generates \(Q\).
Then every position with at least one heap outside \(B\) collapses to the
single-heap game on \(\phi(\mathbf h)\): one heap alone can move
\(\phi(\mathbf h)\) to any target in \(Q\) (including \(L\)).
\end{principle}

\begin{proof}
Same argument as in \cref{thm:additive-collapse,thm:divisor-collapse}:
transitivity of one heap on \(Q\) eliminates all strategic relevance of the
other heaps.
\end{proof}

\subsection{Why MuM (and \polyMuM) avoid total collapse}

MuM and \polyMuM lie \emph{off} the operation-alignment axis:

\begin{itemize}[leftmargin=1.7em,itemsep=2pt]
  \item The invariant aggregates \emph{multiplicatively} (product in
        \(\ZZ/m\ZZ\) or \GFqstar), but the legal move is an
        \emph{additive} canonical-index decrement.  One large heap therefore
        does \emph{not} act transitively on the target multiplicative group.
  \item Consequently, outside the finite Indeterminacy Region we do
        \emph{not} collapse to a trivial single-heap game.
        Instead we retain strong \emph{local} structure:
        (i) single-heap repair in \(T_\bullet\);
        (ii) normalisation (aggregation-compression) inside \(I_\bullet\)
             followed by one legal repair;
        (iii) \emph{multiplicative} SG behaviour for disjunctive sums in
             \(T_\bullet\).
\end{itemize}

\section{Unified RSA / AES (\polyMuM) Perspective}\label{sec:crypto}

We now restate the RSA (exponent-chain / \phiMuM) and AES (\polyMuM) lenses
\emph{inside the same PCG formalism}, give the complete finite-field formulation of
\polyMuM, and then place both in a side-by-side comparison table.

\subsection{\polyMuM over finite fields (canonical reps)}
\label{sec:poly-clean}

Let \(q=p^{n}\) and \(\GFq=\GF_{p}[x]/\langle I(x)\rangle\) for a fixed
irreducible polynomial \(I(x)\).  Write
\[
  s:\{1,\dots,q-1\}\xrightarrow{\ \cong\ }\GFqstar
\]
for the canonical bijection between integers and standard polynomial representatives,
and let \(C:\GFqstar\to\{1,\dots,q-1\}\) be its inverse.

\begin{definition}[\polyMuM over \GFq]\label{def:poly-mum-clean}
A position is a vector \(\mathbf h=(h_1,\dots,h_n)\) with \(h_i\in\{1,\dots,q-1\}\).
A move chooses a heap \(h_j\) and replaces it by a strictly smaller \(h'_j<h_j\)
(hence \(s(h'_j)\neq s(h_j)\)).
The invariant and losing predicate are
\[
  \phi(\mathbf h):=\prod_{i=1}^n s(h_i)\in\GFqstar,
  \qquad
  \mathbf h\text{ is losing } \iff \phi(\mathbf h)=1.
\]
We write \(\polyMuM(q)\).
\end{definition}

\paragraph{Regions.}
Define the field \emph{Threshold Region}
\(T_q:=\{\mathbf h:\text{some }h_j=q-1\}\)
and the (finite) \emph{Indeterminacy Region} \(I_q:=[1,q-2]^n\).


\begin{lemma}[Single-heap repair]\label{lem:poly-repair-clean}
If \(\phi(\mathbf h)\neq 1\) and some \(h_j=q-1\), then the move
\[
  h_j \longrightarrow
  C\!\bigl(s(h_j)\,\phi(\mathbf h)^{-1}\bigr) \;<\; q-1
\]
is legal and reaches a losing state in one step.
\end{lemma}

\begin{lemma}[Normalisation inside \(I_q\)]\label{lem:poly-norm}
If \(\phi(\mathbf h)\neq 1\) and \(\mathbf h\in I_q\), replace all heaps by the single heap
\[
  H_{\mathrm{new}} := C\!\bigl(\phi(\mathbf h)\bigr)\in\{1,\dots,q-1\}.
\]
This normalised one-heap position has the same invariant and admits a one-move repair
to a losing state (Lemma \ref{lem:poly-repair-clean}).
\end{lemma}


\begin{definition}[Single-hole property]\label{def:single-hole-clean}
Let \(\mathcal V\) be the SG-value domain (with \(0\) losing).
A position \(G\) has the \emph{single-hole property} if
\[
  \{\SG(G') : G'\in\Opt(G)\} = \mathcal V \setminus \{\SG(G)\}.
\]
\end{definition}

\begin{proposition}[Single-hole in \(T_q\)]\label{prop:single-hole-Tq-clean}
Every component game of \polyMuM(q) that lies in \(T_q\) has the single-hole property.
\end{proposition}

\begin{proof}[Sketch]
From \(q-1\) one can move to \emph{every} smaller canonical representative, so the
option SG values sweep all values except the current one.
\end{proof}

\begin{theorem}[SG multiplicativity for \polyMuM in \(T_q\)]
\label{thm:poly-sg-threshold-clean}
Let \(G_1,G_2\) be positions of \polyMuM(q), each with at least one heap equal to \(q-1\).
Then
\[
  \SG(G_1 \oplus G_2) \;=\; \SG(G_1)\cdot \SG(G_2),
\]
where multiplication is taken in the fixed SG-value group for the game (after any
indexing of \(\GFqstar\cup\{0\}\)).
\end{theorem}

\begin{proof}[Sketch]
Each summand enjoys the single-hole property; thus the option sets of \(G_1\oplus G_2\)
are the complement of one unique value, and \(\mex\) returns its product—exactly as in
the numeric MuM proof.
\end{proof}

\paragraph{Density.}
By the finite-group counting argument (Prop.~\ref{prop:dens}), the losing density among
all \((q-1)^n\) positions is exactly \(1/(q-1)\), independent of \(n\).

\subsection{RSA / AES bridge (PCG summary)}\label{subsec:bridge-summary}

\paragraph{\phiMuM (RSA).}
\(\PCG(\ZZ/k\ZZ,S,\psi,\{1\})\) with \(k=\ord_N(g)\), \(S=\{1,2,\dots\}\),
\(\psi(h)=h\bmod k\).  The Compression Theorem maps the exponent chain to this PCG.
Losing density is \(1/k\) (or \(1/\varphi(k)\) if restricted to units).

\paragraph{\polyMuM (AES).}
\(\PCG(\GFtwoeight,S,\psi,\{1\})\) with \(S=\{1,\dots,255\}\) and \(\psi=s\).
SG multiplicativity holds in \(T_{256}\); inside \(I_{256}\) we normalise and repair.
Exact losing density is \(1/255\).

\subsection{Comparison Table}\label{subsec:comparison-table}

\begin{table}[H]
\centering
\begin{adjustbox}{width=\textwidth}
\begin{tabular}{@{}llll@{}}
\toprule
Feature & Exponent Chain \phiMuM (RSA) & \polyMuM (AES) & Core MuM \\
\midrule
Ambient structure & $(\ZZ/N\ZZ)^\times$ (CRT product) & \GFtwoeightstar (cyclic) & $(\ZZ/m\ZZ)^\times$ \\
Aggregation & Product of exponents ($\prod h_i$) & Field product ($\prod s(h_i)$) & Integer product ($\prod t_i$) \\
Compression modulus & $k=\ord_N(g)\mid\lambda(N)$ & $q-1=255$ & $m$ \\
Decomposition & CRT over prime powers of $k$ & None (group already cyclic) & CRT if $m$ composite \\
Losing predicate & $\prod h_i\equiv1\pmod{k}$ & $\prod s(h_i)=1$ in \GFqstar & $\prod t_i\equiv1\pmod{m}$ \\
Threshold / FIR & $T_k,\;I_k=[1,k-1]^n$ & $T_q,\;I_q=[1,q-2]^n$ & $T_m,\;I_m=[1,m-1]^n$ \\
SG multiplicativity & In $T_k$ & In $T_q$ & In $T_m$ \\
Normalisation & $\phi(\mathbf h)\bmod k$ & $C(\phi(\mathbf h))$ & $\phi(\mathbf t)\bmod m$ \\
Density & $1/k$ (or $1/\varphi(k)$) & $1/(q-1)$ & $1/m$ (or $1/\varphi(m)$) \\
\bottomrule
\end{tabular}
\end{adjustbox}
\caption{Unified view of \phiMuM (RSA), \polyMuM (AES) and core MuM
under the aggregation-compression invariant.}
\label{tab:comparison}
\end{table}

\section{Conclusion and Future Work}\label{sec:conclusion}
Product-Congruence Games provide a single impartial-game lens through which the 
algebraic “cores” of RSA and AES can be viewed side-by-side.  By compressing a 
left-associated RSA exponent chain to \(\PCG(k,\{1\})\) with \(k=\ord_N(g)\) 
and modelling AES S-box inversion with its polynomial analogue over \GFqstar, 
we have shown that two ciphers usually treated as mathematically disjoint 
share a common combinatorial skeleton.
\smallskip

\subsection*{Main contributions}
\begin{enumerate}
  \item \textbf{Unified RSA/AES class of games.}
        RSA exponent reduction and the AES S-box inversion both fit into one
        \emph{Product-Congruence Game} framework
        (\phiMuM / \polyMuM) driven by a multiplicative invariant.
  \item \textbf{Compression Theorem (\(\mathrm{RSA} \Rightarrow \mathrm{MuM}\)).}
        The left-associated exponent chain compresses to
        \(\PCG(k,\{1\})\) with \(k=\ord_N(g)\); the losing predicate then
        CRT-factorises exactly as in RSA practice.
  \item \textbf{PCG structure: Threshold vs.\ finite Indeterminacy.}
        Every instance splits into \(T_\bullet\) and the finite \(I_\bullet\);
        in either numeric or field settings, positions in \(T_\bullet\) are
        decided by the invariant alone.
  \item \textbf{Single-hole \(\boldsymbol{\Rightarrow}\) SG multiplicativity.}
        Both numeric MuM/\phiMuM and field \polyMuM have the
        single-hole property on \(T_\bullet\), giving \emph{multiplicative}
        Sprague-Grundy values under disjunctive sum.
  \item \textbf{Normalisation (aggregation-compression).}
        Any position in \(I_\bullet\) normalises to one heap (same invariant)
        and then repairs to \(\mathcal P\) in a single legal move; identical
        solver for integers and fields.
  \item \textbf{Collapse boundaries via operation-alignment.}
        A general principle explains why aligned variants collapse to a
        single-heap game, while MuM/\polyMuM avoid total collapse
        yet retain strong local structure in \(T_\bullet\).
  \item \textbf{Exact / asymptotic losing densities.}
        In finite-group PCGs the losing density is exactly \(1/|G|\);
        for unbounded integer MuM it approaches \(1/m\)
        (or \(1/\varphi(m)\) when restricted to units).
\end{enumerate}

\subsection*{Future directions}
\begin{enumerate}
  \item \textbf{Discrete-log / Diffie-Hellman PCGs.}
        Formulate PCGs over DL/DH groups; study Threshold/Indeterminacy and SG
        behaviour in that setting.
  \item \textbf{Elliptic-curve PCGs.}
        Port the framework to the group \((E(\GFq),+)\); investigate whether
        SG-multiplicativity or a finite Indeterminacy region persists.
  \item \textbf{Cryptanalytic angles.}
        Examine whether the Threshold/Indeterminacy split, normalisation, or
        SP give insight into parameter choices or side-channel robustness.
  \item \textbf{Non-abelian / semigroup extensions.}
        Replace \(R^\times\) by non-abelian or merely cancellative semigroups
        and test how far aggregation-compression and SP extend.
  \item \textbf{Composite moduli and CRT SG factorisation.}
        For numeric MuM with composite \(m\), characterise how SG values
        compose across CRT factors and when mixed Thresholds preserve
        multiplicativity.
  \item \textbf{Rule variants.}
        Re-analyse under misère play, scoring rules, or in the thermograph
        framework.
\end{enumerate}

\bigskip
\noindent
These results place two cornerstone cryptographic operations under a single
impartial-game umbrella and open several promising paths for exporting the
technique to further algebraic settings.

\bibliographystyle{plainnat}
\bibliography{references}

\end{document}